\documentclass[runningheads]{llncs}

\usepackage{enumitem}
\usepackage{amssymb}
\usepackage{amsmath}
\usepackage{graphicx}
\usepackage{verbatim}
\usepackage{color}
\usepackage{ae}
\usepackage[ruled]{algorithm}
\usepackage{algorithmicx}
\usepackage{algpseudocode}
\usepackage[breaklinks,bookmarks=false]{hyperref}

\DeclareMathOperator{\LCPref}{LCPref}
\DeclareMathOperator{\LCSuf}{LCSuf}

\DeclareMathOperator{\led}{led}
\DeclareMathOperator{\ed}{ed}
\DeclareMathOperator{\leftval}{leftval}

\DeclareMathOperator{\TotalLen}{TotalLen}
\DeclareMathOperator{\mism}{mism}

\newlength\mylen
\newlist{mycases}{enumerate}{1}
\setlist[mycases,1]{label=\textbf{Case~\arabic*.}, 
  labelwidth=\dimexpr-\mylen-\labelsep\relax,leftmargin=0pt,align=right}
\newcommand\mynobreakpar{\par\nobreak\@afterheading}

\newcommand{\proc}[1]{\textnormal{\scshape#1}}

\begin{document}

\title{Beating $\mathcal{O}(nm)$ in approximate LZW-compressed pattern matching}

\author{Pawe\l{} Gawrychowski\inst{1} and Damian Straszak\inst{2}}
\authorrunning{}
\institute{
Max-Planck-Institut f\"{u}r Informatik, Saarbr\"ucken, Germany,\\ \email{gawry@cs.uni.wroc.pl}
\and
Institute of Computer Science, University of Wroc{\l}aw, Poland,\\ \email{damian.straszak@gmail.com}
}

\maketitle

\begin{abstract}
Given an LZW/LZ78 compressed text, we want to find an approximate occurrence of a given pattern of length $m$. The goal is to achieve time complexity
depending on the size $n$ of the compressed representation of the text instead of its length. We consider two specific definitions of approximate matching, namely the Hamming distance and the edit distance, and show how to achieve $\mathcal{O}(n\sqrt{m}k^{2})$ and $\mathcal{O}(n\sqrt{m}k^{3})$ running time, respectively, where $k$ is the bound on the distance. Both
algorithms use just linear space.
Even for very small values of $k$, the best previously known solutions required $\Omega(nm)$ time. Our main contribution is applying a periodicity-based argument in a way that is computationally effective even if we need to operate on a compressed representation of a string, while the previous solutions were either based on a dynamic programming, or a black-box application of tools developed for uncompressed strings.
\end{abstract}

\section{Introduction}

Pattern matching, which is the question of locating an occurrence of a given pattern in a text, is the most natural task as far as processing text data is concerned. Virtually any programming language contains a more or less efficient procedure for solving this problem, and any text processing application, including the widely available \texttt{grep} utility, gives users the means of solving it.

While exact pattern matching is well-understood, and in particular many linear time solutions are known~\cite{Jewels}, it seems that in its approximate version, where one asks for occurrences that are similar to a given pattern, there is still some room for improvement. Two most natural versions of the question are pattern matching with errors, where one ask for a substring of the text with small edit distance to the pattern, and pattern matching with mismatches, where one is interested in a substring with small Hamming distance to the pattern, which is simply the number of mismatched characters. It is known that if $N$ is the length of the text and $k$ is the number of allowed errors or mismatches, both problems can be solved in $\mathcal{O}(Nk)$ time~\cite{LandauMismatches,Landau}, and in fact the complexity for the latter version can be improved to $\mathcal{O}(N\sqrt{k\log k})$~\cite{AmirMismatches}. Under the natural assumption that the value of $k$ is small, one can do even better, and solve the problems in $\mathcal{O}(N+\frac{Nk^{4}}{m})$~\cite{ColeHariharan} and $\mathcal{O}((N+\frac{Nk^{3}}{m})\log k)$~\cite{AmirMismatches} time complexity, respectively, which might be linear in $N$ if $k$ is small enough. Unfortunately, in some cases even a linear time complexity might be not good enough. This is the case when we are talking about large collections of repetitive data stored in
a compressed form. Then the length of the text $N$ might be substantially larger than the size $n$ of its actual representation, and the goal is to achieve a running time depending on $n$, not $N$. Whether achieving such goal is possible clearly depends on the power of the compression method. In this paper we focus on the LZW/LZ78 compression~\cite{LZW,LZ78}, which is not as powerful as the more general LZ77 method, but still has some nice theoretical properties, and is used in real-world applications. It is known that exact LZW-compressed pattern matching can be solved very efficiently~\cite{Amir,GawrychowskiLZW}, even in the fully compressed version, where both the text and the pattern are LZW-compressed~\cite{GawrychowskiFully}. The obvious question is how efficiently can we solve approximate LZW-compressed pattern matching?

The best previously known solution by K{\"a}rkk{\"a}inen, Navarro, and Ukkonen~\cite{Juha}, locates all $occ$ occurrences with up to $k$ errors using $\mathcal{O}(nmk+occ)$ time and $\mathcal{O}(nmk)$ space. More precisely, it outputs all ending positions $j$ such that there is $i$ for which the edit distance between $t[i..j]$ and $p$ is at most $k$. In some cases, this time bound can be decreased using the idea of Bille, Fagerberg, and G{\o}rtz~\cite{Bille}, who presented a way to translate all uncompressed pattern matching bounds into the compressed setting. Their approach works for both the edit and Hamming distance, and by plugging the best known uncompressed pattern matching solutions, we can get:
\begin{enumerate}
\item $\mathcal{O}(nmk+occ)$ time and $\mathcal{O}(\frac{n}{mk}+m+occ)$ space solution for the edit distance,
\item $\mathcal{O}(nk^{4}+nm+occ)$ time and $\mathcal{O}(\frac{n}{k^{4}+m}+m+occ)$ space solution for the edit distance,
\item $\mathcal{O}(nk^{3}\log k+nm\log k+occ)$ time and $\mathcal{O}(\frac{n}{(k^{3}+m)\log k}+m+occ)$ space solution for the Hamming distance.
\end{enumerate}
While the space complexity of the resulting algorithms is small, even for constant values of $k$ the time complexity is $\Omega(nm)$, and in fact this is an inherent shortcoming of the approach: the best we can hope for is $\mathcal{O}(nm)$ for sufficiently small values of $k$, say, $k=\mathcal{O}(m^{1/3})$. In this paper we show that in fact this barrier can be broken. We prove that for the Hamming distance, running time of $\mathcal{O}(n\sqrt{m}k^{2})$ is possible, which for $k=o(m^{1/4})$ is $o(nm)$. Then we show how to extend the algorithm by building on the ideas of Cole and Hariharan~\cite{ColeHariharan}, and achieve $\mathcal{O}(n\sqrt{m}k^{3})$ for the edit distance. Both algorithms use $O(n+m)$ space. For the sake of making the description clear, we concentrate on the question of detecting just one occurrence, but our algorithms generalize to generating all of them in the left-to-right order.

Some of our methods are based on the concepts first used by Cole and Hariharan~\cite{ColeHariharan}, and later by Amir, Lewenstein, and Porat~\cite{AmirMismatches}. We would like to
stress out that applying them in the compressed setting is not just a trivial exercise, and creates new challenges. For instance, verifying whether a given position corresponds
to an occurrence with no more than $k$ mismatches in $\mathcal{O}(k)$ time is straightforward in the uncompressed setting using the suffix tree, but in our case requires some additional ideas.

We start with some basic tools in Sections~\ref{section:preliminaries} and~\ref{section:further}. Then we distinguish between two types of matches, called internal and crossing. Detecting the former is relatively straightforward in both versions. To detect the latter, we reduce the question to a problem that is easier to work with, which we call pattern matching in pc-strings, in Section~\ref{section:reduction}. To solve pattern matching with mismatches in pc-strings, we distinguish between
two cases depending on how periodic the pattern is. For this we apply the concept of breaks, heavily used in the previous papers on approximate pattern matching. In our case, though,
we are looking at $z$-breaks for some value of $z$ that will be specified in the very end. If there are many such breaks, or in other words the pattern is not very repetitive, we
can solve the problem by reducing to a generalization of (exact) compressed pattern matching with multiple patterns as shown in Section~\ref{section:weakly}. Otherwise, the pattern is {\it highly periodic},
and the situation is more complicated. In Section~\ref{section:highly} we show how to exploit the regular structure of such pattern to construct an efficient algorithm.
More precisely, we construct a small set of candidates for a potential occurrence, and verify them one by one. Then
in Section~\ref{section:faster}, which is the most technical part of the paper, we speed up the method using a new technique which considers all candidates in a more global manner instead of operating on them separately.
Finally, in Section~\ref{section:errors} we generalize the solution to solve the version with errors.

\section{Preliminaries}
\label{section:preliminaries}

We are given a text $t[1..N]$ and a pattern $p[1..m]$, both are strings over an integer alphabet $\Sigma$. We assume that $m\leq N$ and $\Sigma=\{1,2,\ldots,N\}$. The pattern is given explicitly, but the text is described implicitly using the LZW/LZ78 compression scheme. Such scheme is defined as follows: we partition the text into $n$ disjoint fragments $t=z_{1} z_{2} \ldots z_{n}$, where each fragment $z_{i}$ is either a single letter, i.e., $z_{i}=c$, or a word of the form $z_{i}=z_{j} c$, where $j<i$. The fragments $z_{i}$ are usually called the {\it codewords}, and because their set is closed under taking prefixes, we may represent it as a trie, which will be further denoted by $T$. Depending on how we choose the partition and encode the codewords, we get different concrete compression methods, say LZW or LZ78. Our methods do not depend on such technicalities as long as we are given $T$ and the text is described as a list of pointers to the nodes of $T$ representing the successive fragments. From now on whenever we mention compression, we mean such representation.

The Hamming distance between two strings of the same length is simply the number of positions where their corresponding characters differ. The edit distance $\ed(s,t)$ is the minimal number of operations necessary to transform $s$ into $t$, where an operation is an insertion, replacement, or removal of a character.

The first problem we consider is {\it compressed pattern matching with mismatches}, where we are given a compressed representation of a text $t$, a pattern $p$, and a positive integer $k$. We want to find $i$ such that the Hamming distance between $t[i..i+m-1]$ and the pattern is at most $k$. 
We also consider {\it compressed pattern matching with errors}, where the goal is to find $i$ and $j$ such that the edit distance between $t[i..j]$ and $p$ is at most $k$.
Our solutions generalize to generating all occurrences, where we are asked to report the ending positions of all matches (notice that in pattern matching with errors, there might
be multiple occurrences ending at the same position, and then we want to report such position just once), but we concentrate on the easier to describe version with just one
occurrence.

To efficiently operate on the compressed text and the pattern we need a number of data structures. We assume that the suffix array of the pattern is available. A description of this data structure, together with a simple linear time construction algorithm, can be found in~\cite{Karkkainen}. Notice that it assumes that the alphabet is of polynomial size, which is the case here as $N\leq n^{2}$. By adding a constant time range minimum query structure~\cite{Bender}, we get the following useful primitive.
\begin{lemma}\label{lemma:lcapattern}
A string can be preprocessed in linear time so that given its any two fragments we can find their longest common prefix and longest common suffix in constant time.
\end{lemma}
It may seem surprising, but the problem of efficiently retrieving a single letter given a position in the text is rather nontrivial. By ``position in the text" we mean a pair $(i,j)$
which stands for ``$j$-th letter of $z_i$". This problem can be reduced to accessing the $k$-th ancestor of a node in a tree. A straightforward algorithm gives us $\mathcal{O}(n \log n)$ preprocessing time and $\mathcal{O}(\log n)$ time per query, with $n$ being the size of the tree as in our setting, and a much better solution is known~\cite{AlstrupAncestor}, giving the following tool.
\begin{lemma}\label{lemma:lettertext}
We can preprocess the text in linear time so that given any position there we can retrieve the corresponding letter in constant time.
\end{lemma}
The next two lemmas will provide methods for comparing substrings of the text with substrings of the pattern. The usual solution in such cases is to build suffix tree of the text concatenated with the pattern. However, this would require decompressing the text and result in a structure of size $\Omega(N)$, which is unacceptable. Hence we need another method.

From now on, let $T_p$ denote a trie representing the set of strings $\{z_1, z_2, ..., z_n\}\cup \{p[1..k]:1\leq k\leq m\}$, in other words $T_p$ is $T$ extended by one path representing the whole pattern. The size of $T_p$ is $|T_p|=\mathcal{O}(n+m)$.

\begin{figure}[t]
\centering
\includegraphics[scale=0.7]{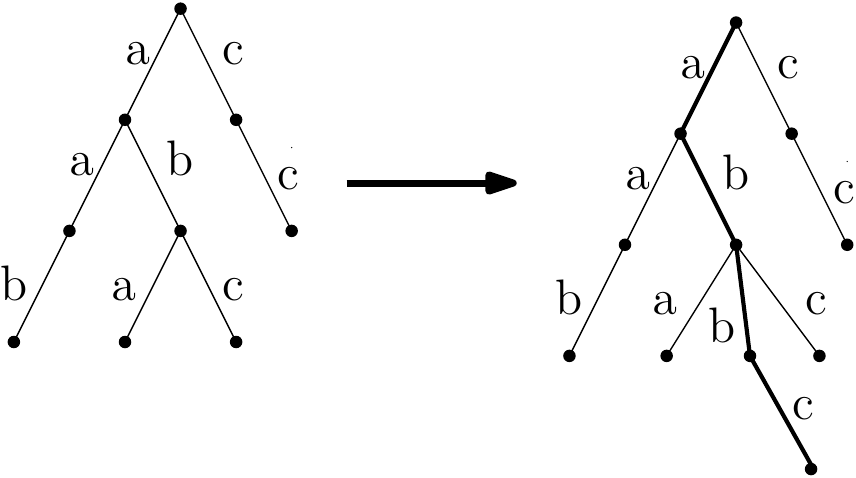}
\caption{Constructing $T_p$ from $T$ and $p=abbc$.}
\label{figure:trie_pattern}
\end{figure}

\begin{definition}
Let $S$ be a trie of strings. We say that $s$ is a chunk (in $S$) if $s$ is a subword of some root-to-leaf path in $S$. We represent a chunk as a pair consisting of a node in $S$ corresponding to the last letter in $s$ and the length $|s|$.
\end{definition}

$\LCSuf$ is the longest common suffix and $\LCPref$ is the longest prefix of given two strings. We are interested in computing them for any two chunks.

\begin{lemma}\label{lemma:LCSuf}
Given any trie $S$ we can preprocess it in linear time so that the $\LCSuf$ of any two chunks can be computed in constant time.
\end{lemma}

\begin{proof}
Observe that we may concentrate only on the chunks beginning in the root. Let $W$ be the set of all such chunks. Imagine now a compacted trie $S^R$ representing all words in the set $W^R=\{w^R:w\in W\}$, where $w^R$ stands for the reversed string $w$. It may be seen that the size of the tree is $\mathcal{O}(|S|)$ because each $w\in W$ adds at most one leaf. $S^R$ is an interesting generalization of a suffix tree of a regular string, and as in the basic case it may be constructed in linear time assuming an integer alphabet of polynomial size~\cite{Shibuya99}. After building such a tree, we augment it with a linear size structure allowing constant time LCA queries~\cite{Bender}. This, along with pointers between nodes of $S$ and the corresponding nodes of $S^R$, gives the claimed result.
\qed
\end{proof}

\begin{lemma}\label{lemma:LCPref}
We can preprocess the text and the pattern in $\mathcal{O}(m+n)$ time so that given a chunk in $T$ and a subword of the pattern we can find their $\LCSuf$ in $\mathcal{O}(1)$ time and $\LCPref$ in $\mathcal{O}(\log m)$ time.
\end{lemma}
\begin{proof}
We use the previous lemma for $S=T_p$. We get immediately the result for $\LCSuf$. Now consider an $\LCPref$ query: we use binary search together with $\LCSuf$ queries to get the answer in $\mathcal{O}(\log m)$ time. Note that in such a procedure we need often to find the $k$-th ancestor in the tree, but we already know by Lemma \ref{lemma:lettertext} that it can be done in constant time.
\qed
\end{proof}

We need also some basic concepts from combinatorics on words. $\alpha$ is a period of a string $s$ if $s[i]=s[i+\alpha]$ holds for every $i=1,2,\ldots,|s|-\alpha$, or in other words
we can write $s=w^i u$, where $|w|=\alpha$ and $u\neq w$ is a prefix of $w$. The smallest such $\alpha$
is called {\bf the} period of $s$. If the period of $s$ is at most $\frac{|s|}{2}$, $s$ is periodic, and otherwise we call it a break, or $|s|$-break. A word is primitive if it cannot be
represented as a nontrivial power of some other word. For every word $s$, there exists its unique cyclic shift $s'$ which is lexicographically smallest, and we call $s'$
the cyclic representative of $s$. For a periodic $s$, the cyclic representative of $w$ corresponding to the period of $s$ is called the canonical period of $s$. One of the basic
results concerning periods is the periodicity lemma, which says that if $q$ and $q'$ are both periods of $s$, and $q+q'\leq |s|$, so is $\gcd(q,q')$.

\section{Further preprocessing}
\label{section:further}

From now on we fix $k$ to be the number of allowed mismatches (errors) in our problem. We will say in short that the pattern matches at some position in the text if the Hamming distance (or the edit distance) between the pattern and the fragment of the text starting at this position is at most $k$. It is natural to distinguish between two types of matches: {\it internal matches} (the pattern lies fully within a single codeword) and  {\it crossing matches} (the pattern crosses some boundary between two codewords). We are now going to show that one can easily find all internal matches.

\begin{lemma}\label{lemma:internal_matches}
In case of pattern matching with mismatches we can find all internal matches in $\mathcal{O}(nk+m)$ time and if needed report all of them in $\mathcal{O}(1)$ time per occurrence.
\end{lemma}
\begin{proof}
Each internal match can be seen as a chunk in $T$ of length $m$ (but of course one such chunk may correspond to many places in the text). For every such chunk we verify whether its Hamming distance with the pattern is at most $k$. This can be done in $\mathcal{O}(k)$ time using constant time $\LCSuf$ queries to jump over whole fragments with no mismatches, which is the standard method used in the uncompressed setting~\cite{LandauMismatches}. In total there are $\mathcal{O}(n)$ chunks to verify, which gives $\mathcal{O}(nk+m)$ time. Reporting all occurrences is straightforward if we know which chunks match the pattern.
\qed
\end{proof}

\begin{lemma}\label{lemma:internal_matches_errors}
In case of pattern matching with errors we can find all internal matches in $\mathcal{O}(nk^2+m)$ time and if needed report all of them in $\mathcal{O}(1)$ time per occurrence.
\end{lemma}
\begin{proof}
We use the same method as in the proof of Lemma~\ref{lemma:internal_matches}, but to verify a match we apply the Landau-Vishkin algorithm~\cite{Landau}, which computes the edit distance using a clever dynamic programming in $\mathcal{O}(k^{2})$ time (see Section 5 of~\cite{ColeHariharan}), assuming we can answer any $\LCSuf$ query in constant time (in the original formulation, $\LCPref$ queries are used, but we can simply pretend that the text and the pattern are reversed). More precisely, given two strings, it can be used to compute all prefixes of the former whose edit distance to the latter is at most $k$. This gives us, for each node of $T$, the chunks ending there and corresponding to a match in $\mathcal{O}(nk^{2}+m)$ time. This can be modified to report all occurrences in a straightforward manner. 
\qed
\end{proof}

The situation with crossing matches is much more complicated. In this case the pattern crosses at least one boundary between two codewords, and it may cross a lot of them, which seems hard to deal with. Anyway, it suffices to iterate over all $n-1$ boundaries and for each of them find all matches that cross it. After fixing such a boundary, we may concentrate only on a window of length $2m$ containing $m$ characters to the left and $m$ to the right. Problems arise when there are many very short codewords in some fragment of the text, because in such a case all boundaries in this fragment will create windows containing lots of codewords. This is one of the obstacles we need to tackle to construct an efficient algorithm.

We want to make now one technical assumption, which simplifies significantly some definitions and the description of the algorithm. Namely, we will assume that each letter appearing in text, appears also in the pattern. Our algorithms work in the general case after minor modifications.

The notion of a pc-string will play the main role in the rest of the paper. Note that the definition changes slightly when we want to move from mismatches to errors. Nevertheless, the change is very small, so we prefer to have just one common definition, and keep in mind that its meaning depends on the variant.

\begin{definition}
Let $p$ be a pattern and $f$ be a string. We say that $f=v_1v_2...v_l$ is a pattern-compressed-string, in short pc-string, if: 
\begin{enumerate}
\item $|f|\leq 2m$ ($|f| \leq 2m+2k$ when we are dealing with errors) and $l\leq 4k+5$,
\item $v_i$ is a factor of $p$, for $i=1,2,\dots,l$,
\item $v_iv_{i+1}$ is not a factor of $p$, for $i=1,2,\dots,l-1$.
\end{enumerate}
We represent such string as a list $(a_1,b_1),(a_2,b_2),...,(a_l,b_l)$, where $v_i=p[a_i..b_i]$.
\end{definition}

Pc-strings are very convenient to deal with. The fact that no $v_iv_{i+1}$ appears in $p$ as a substring allows us to answer LCPref and LCSuf query between a subword of $f$ and a subword of the pattern in constant time, as each result of such a query overlaps at most 3 $v_i$'s, so we need at most 3 queries between factors of $p$. This also implies the
following proposition.


\begin{proposition}\label{proposition:verify_match_mismatches}
Given a position in a pc-string $f$, we can verify whether the alignment of the pattern at this position results in a match in $\mathcal{O}(k)$ time.
\end{proposition}
\begin{proof}
We already know that performing $\LCPref$ queries takes constant time. Now the result follows, again, by using constant time $\LCPref$ queries to jump over whole fragments with no mismatches. Now there is one additional detail, though. To answer such query in constant time, we need to maintain the corresponding current position in $f$, or more precisely, the current $v_{i}$ and the current letter there. This is easily done, as during the computation we only move to the right, and only to either $v_{i+1}$ or $v_{i+2}$, hence the current position can be updated in constant time.
\qed
\end{proof}

It turns out that finding matches crossing a fixed boundary can be reduced to one instance of pattern matching with mismatches or errors in a pc-string. This reduction is presented with details in the next section. Here we conclude it by the following theorem:

\begin{theorem}\label{theorem:reduce_to_PCstrings}
Suppose we have an algorithm solving pattern matching with $k$ mismatches (errors) in pc-strings in $T_{PC}(m)$ time. Then we can solve pattern matching with $k$ mismatches (errors) in LZW-compressed text in $\mathcal{O}(nk\log^2 m +m+n\cdot T_{PC}(m))$ ($\mathcal{O}(nk^2+nk\log^2 m +m+n\cdot T_{PC}(m))$) time.
\end{theorem}

In order to solve pattern matching in a pc-string, we will extensively use a certain preprocessing of the pattern. This preprocessing takes $\mathcal{O}(m)$ and is performed just once in the whole algorithm, not every time we get an instance of pattern matching in a pc-string. Hence we usually do not 
include this time in the statements of the lemmas.

\section{Reducing to pc-strings}
\label{section:reduction}

In this section we show a method for reducing the problem of finding matches crossing some boundary to pattern matching with mismatches in a pc-string. We will focus on the version with mismatches, as the version with errors requires just very minor modifications.

The main technical tool we are going to use in this section is a PH-decomposition.

\begin{definition}
Given a pattern $p$ and a compressed text $t=z_1z_2...z_n$ we say that $t=w_1w_1...w_d$ is a $PH$-decomposition if the following conditions are met.
\begin{enumerate}
\item{Each block $w_i$ is either a factor of the pattern (we say it is of type $P$) or is fully contained within a single codeword (we say it is a ``hole" or of type $H$).}
\item{Each codeword contains at most one hole.}
\item{For each two consecutive blocks $w_i$, $w_{i+1}$ of type $P$, their concatenation $w_iw_{i+1}$ does not appear as a subword in the pattern.}
\item{If $w_i$ is a hole, then the whole $w_{i-(2k+3)}w_{i-(2k+2)}\ldots w_i\ldots w_{i+(2k+3)}$ lies within a single codeword and all $w_{i-(2k+3)}, \ldots, w_{i-1}, w_{i+1}, \ldots, w_{i+(2k+3)}$ are of type $P$. }
\end{enumerate}
\end{definition}
Figure~\ref{figure:PH-decomp} demonstrates an important property of a PH-decomposition. Between a hole and the boundaries of the block it originates from there are always at least $2k+3$ full P-blocks.

\begin{figure}[t]
\includegraphics[width=\textwidth]{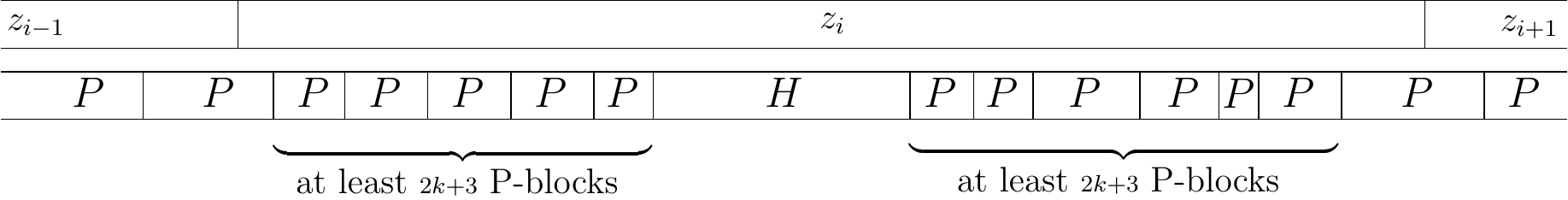}
\caption{An important property of a PH-decomposition.}
\label{figure:PH-decomp}
\end{figure}

We will now give an algorithm for finding a $PH$-decomposition of the text. It is useful to first show the following lemma.

\begin{lemma}\label{lemma:longfactor}
Given a chunk $s$ in $T$ we can:
\begin{enumerate}
\item{find the longest prefix of $s$ which appears as a subword in $p$ (and locate this subword in p) in $\mathcal{O}(\log^2 m)$} time,
\item{find the longest suffix of $s$ which appears as a subword in $p$ (and locate this subword in p) in $\mathcal{O}(\log m)$} time.
\end{enumerate}
\end{lemma}

\begin{proof}
Let us show the first part. We are able to compare lexicographically $s$ with a suffix of the pattern in $\mathcal{O}(\log m)$ time, since we can use one $\LCPref$ query and compare the next letters. Imagine the suffix array of $p$, and suppose we want to insert $s$ into this array maintaining the lexicographical order. We can find the right place for $s$ in $\mathcal{O}(\log^2 m)$ time by binary search. Then the immediate predecessor or successor of $s$ will have the maximum possible $\LCPref$ with $s$, and this will correspond to the maximal prefix of $s$ which is a subword of $p$. 

The second part of the lemma holds because we can answer $\LCSuf$ queries in constant time.
\qed
\end{proof}

\begin{lemma}
A $PH$-decomposition of the text can be found in $\mathcal{O}(nk\log^2m+m)$ time.
\end{lemma}

\begin{proof}
The idea is to first work with single codewords and decompose them partially, then take all these decompositions and merge them together into one $PH$-decomposition of the whole text.

Let us consider a codeword $z$. We want to represent $z$ in one of the following forms:
\begin{enumerate}
\item $z=v_1v_2...v_l$, where each $v_i$ is a factor of the pattern, $l\leq 4k+7$, and none of the words $v_iv_{i+1}$ appears in $p$ as a subword,
\item $z=v_1v_2...v_{2k+4}hv'_{2k+4}v'_{2k+3}...v'_{1}$ where each $v_i$ and each $v'_i$ is a factor of the pattern, and no $v_iv_{i+1}$ nor $v'_{i+1}v'_{i}$ appears in $p$ as a subword; here $h$ can be seen as a ``hole".
\end{enumerate}
We go from left to right through z. First we use Lemma~\ref{lemma:longfactor} to find $v_1$, the longest prefix of $z$ which appears in $p$ as a subword. We then erase $v_1$ from the beginning of $z$ and proceed similarly to find $v_2$, and so on. If this procedure stops after at most $4k+7$ steps, meaning there is nothing left from $z$, we are done, as we have got the first possible form of $z$. In the opposite case, we stop and do the same starting from the end of $z$ and going from right to left. We perform $2k+4$ such steps getting $v_1', v_2', ..., v_{2k+4}'$. Clearly, we can get the second form by taking the first $2k+4$ $v_i$'s, all $v_i'$'s, and setting $h$ to be the remaining middle part of $z$.

Now take all these representations and merge them together. We declare all $v_i$'s and $v_i'$'s to be of type $P$ and all $h$'s to be of type $H$. Note that we are almost done, as all conditions in the definition of a $PH$-decomposition except (maybe) the third are fulfilled. The violation of the third rule can occur in places where the representations were glued together. To fix this we repeat the following procedure until the third condition holds: take two consecutive blocks $w_i$ and $w_{i+1}$ of type $P$ such that $w_iw_{i+1}$ appears as a subword in $p$, replace $w_i, w_{i+1}$ by one $P$-block $w_iw_{i+1}$. In order to do this, we need a way to check whether a word of the form $w_iw_{i+1}$ appears as a subword in $p$. This can be done by locating the position of $w_{i}w_{i+1}$ in the suffix array of $p$ using binary search in $\mathcal{O}(\log m)$ time, as in the proof of Lemma~\ref{lemma:longfactor}.

In total we produce at most $\mathcal{O}(nk)$ blocks, and spend $\mathcal{O}(\log^2m)$ time per block. Thus the claimed time bound.
\qed
\end{proof}

Suppose now we have a $PH$-decomposition of the text $t=w_1w_1...w_d$. Thanks to this, we can simplify the problem of finding crossing matches. Fix a boundary between two codewords and suppose it is either inside or just before a $P$-block $w_i$. From the definition, all blocks $w_{i-(2k+2)},w_{i-(2k+1)}, \ldots, w_{i+(2k+2)}$ are of type $P$.  Furthermore, we claim that every match crossing our fixed boundary lies within the fragment $w_{i-(2k+2)}w_{i-(2k+1)}...w_{i+(2k+2)}$ of the text. This is because otherwise it either ends after $w_{i+(2k+2)}$ or starts before $w_{i-(2k+2)}$, so there are $2k+2$ consecutive $P$-blocks such that their concatenation match with at most $k$ mismatches with some subword of the pattern. But then some two consecutive blocks have to match exactly with some subword of the pattern, which contradicts the definition of the $PH$-decomposition.

Now a string $s$ consisting of the at most $4k+5$ $P$-blocks considered above is almost a pc-string. We only need to trim it so that its length does not exceed $2m$ ($m$ to the left and $m$ to the right from boundary). After such trimming the first or the last block might become shorter, and it might be necessary to merge such incomplete block with the neighbour, which can be done in $\mathcal{O}(\log m)$ time as in the above proof. This concludes our reduction.

If we just want to find the first occurrence of the pattern, then we process the boundaries one by one and solve matching in the corresponding (at most) $2m$-length windows of the text. However, if our aim is to report all the occurrences of the pattern, we need to make sure no match is reported multiple times. For this just trim the considered windows so that they overlap on at most $m-1$ positions. Note also that when performing the reduction carefully we only need to store $\mathcal{O}(k)$ pattern factors at a time, so it won't affect the space complexity.

\section{Detecting matches in pc-strings}
\label{section:weakly}

In this section we concentrate on the version with mismatches and present an efficient algorithm for detecting matches in a pc-string (recall that by Theorem~\ref{theorem:reduce_to_PCstrings} this solves the problem of finding matches in a compressed text).
We distinguish two cases depending on the ``level of periodicity" of the pattern. Let $z\geq 3$ be a parameter to be fixed later. We find in $p$ as many disjoint $z$-breaks as possible. If there are just a few such breaks, the pattern can be seen as {\it highly periodic}.

For finding the maximum number of disjoint breaks in the pattern, we use the \proc{Find-breaks} procedure of Cole and Hariharan~\cite{ColeHariharan}. We prove its correctness for the sake of completeness.

\begin{algorithm}[t]
\caption{\proc{Find-breaks}(p)}
\begin{algorithmic}[1]
\State $i\gets 1$
\While{$i\leq |p|-z+1$}
\If{the period of $p[i..i+z-1]$ exceeds $\frac{z}{2}$}
\State report $z$-break $p[i..i+z-1]$
\State $i \gets i+z$
\Else
\State $i\gets i+1$
\EndIf
\EndWhile
\end{algorithmic}
\end{algorithm}

\begin{lemma}\label{lemma:finding_breaks}
\proc{Find-breaks}$\mathrm{(p)}$ finds the maximum possible number of disjoint $z$-breaks in $p$. Moreover, it can be implemented in $\mathcal{O}(m)$ time.
\end{lemma}
\begin{proof}
To show the correctness of the procedure, we proceed by induction on $|p|$.  \proc{Find-breaks} locates the leftmost $z$-break, and then continues on 
its right, hence it is enough to argue that it is always possible to choose the maximum number of disjoint $z$-breaks with the first break being the
leftmost break. But this is obvious, because we can always move the first break in the solution to the left.

Let us now concentrate on the time complexity. After the algorithm terminates, the pattern is of the form $p=s_1b_1s_2b_2...b_{r}s_{r+1}$, where each $b_i$ is a break and each $s_i$ is $q$-periodic for some $q\leq \frac{z}{2}$.  We will show how to handle a portion $w_i=s_ib_i$ in $\mathcal{O}(|s_ib_i|)$ time. We start with the window containing the length-$z$ prefix of $w_i$, we find its period $q$ in $\mathcal{O}(z)$ time. If $q>\frac{z}{2}$ then stop. Otherwise we keep adding letters to the end, one by one, checking only whether the new character is the same as the one $q$ positions earlier. First time this is not the case, we stop and mark the length-$z$ suffix of the considered portion as a break $b$. This works because if $b$ had period $q'\leq \frac{z}{2}$ (necessarily $q\neq q'$) then, by the periodicity lemma, $b$ with last character erased would have period at most $\gcd(q,q')<q$, which is impossible.
\qed
\end{proof}

We consider now the case when $p$ contains at least $2k$ disjoint $z$-breaks. It turns out that in such a case we can discard most of the starting positions, and then verify all remaining candidates separately.

\begin{lemma}\label{lemma:sparsify_matches}
Let $f$ be a text of length $2m$. Assume that the pattern $p$ contains at least $2k$ disjoint $z$-breaks. Then there are at most $\mathcal{O}(\frac{m}{z})$ matches (with $k$ mismatches) of $p$ in $f$.
\end{lemma}

\begin{proof}
Choose $2k$ disjoint occurrences of breaks in the pattern. Let $b_1, b_2, ..., b_r$ be all pairwise different breaks among them, with $b_i$ occurring $x_i$ times, so $\sum_{i=1}^{r}x_i=2k$. Consider one break $b_i$, and denote the positions of the disjoint occurrences of $b_i$ in $p$ by $o_1, o_2, ..., o_{x_i}$. For each occurrence of $b_i$ in the text, say at position $q$, we add a mark to all positions $q-o_{1}+1, q-o_2+1, \ldots, q-o_{x_i}+1$ within the text. Since the distance between two different occurrences of $b_i$ in the text is at least $\frac{z}{2}$ (because if they were closer it would imply $b_i$ has shorter period than $\frac{z}{2}$), there will be at most $x_i\frac{2m}{z}$ marks caused by $b_i$. So all in all, there will be at most $\sum_{i=1}^{r}x_i\frac{2m}{z}=\frac{4km}{z}$ marks. Consider now a position in the text where $p$ matches with at most $k$ mismatches. At least $k$ of the $2k$ breaks have to match exactly, so we have at least $k$ marks there. But there are only at most $\frac{4m}{z}$ positions with at least $k$ marks, so the lemma follows.
\qed
\end{proof}

This lemma is very useful, but it does not give a method to find all these $\mathcal{O}(\frac{m}{z})$ positions. For this need to locate all occurrences in $f$ of up to $2k$ pattern breaks. We cannot simply use the usual multiple pattern matching algorithm, because it would cost $\Omega(m)$ time, which is too much. However, we know that there are at most $\mathcal{O}(\frac{km}{z})$ occurrences of these breaks in $f$. This fact, combined with an efficient algorithm for multiple pattern matching in a pc-string, which is an adaptation
of the method of Gawrychowski~\cite{GawrychowskiMultipleLZW}, gives a solution.

\begin{lemma}\label{lemma:break_matching}
Suppose $p$ is a pattern and $b_1, b_2, ..., b_r$ is a fixed collection of its disjoint and pairwise different $z$-breaks. We can preprocess the pattern and the collection of breaks
in $\mathcal{O}(m)$ time, so that later given any pc-string $f=v_1v_2...v_l$ we can find all $occ$ occurrences of $b_1, b_2, ..., b_r$ in $f$ in $\mathcal{O}(l\log m+occ)$ time.
\end{lemma}

\begin{proof}
We assume that we have a sorted array of all suffixes and all prefixes of all breaks. We also assume that all occurrences of $b_1,b_2,...,b_r$ in $p$ have been found using the Aho-Corasick automaton~\cite{AhoCorasick}, and that we have an array storing for each $j$ the leftmost occurrence of any break from the collection in $p[j..m]$.
As a byproduct of generating all occurrences, we also get an array storing for each $j$ the longest prefix of $p[j..m]$ which is a prefix of some break.
By reversing the pattern and running the automaton again, we can also get an array storing for each $j$ the longest suffix of $p[1..j]$ which is a suffix of some break. Additionally, we organize the prefixes of all breaks into a {\it prefix tree}, where the parent of $b_{i}[1..j]$ is the longest prefix of some break which is a proper suffix of $b_{i}[1..j]$. Such tree can be constructed in linear time using a single scan over the sorted array of all reversed prefixes of all breaks. The prefix tree is augmented with a constant time level ancestor structure.

The preprocessing is done in $\mathcal{O}(m)$ time because the total length of all breaks and the total number of their occurrences is at most $m$.
Now we will see how to adapt the algorithm given in~\cite{GawrychowskiMultipleLZW} to work in our case. One can see, that it is enough to show how to perform the following operations:
\begin{enumerate}
\item given some $v_i$ find its longest suffix (prefix), which is a prefix (suffix) of some $b_j$,
\item given some $v_i$ check whether it is a subword of some $b_j$ (and locate this subword),
\item given some $v_i$ find all occurrences of $b_j$'s inside it.
\end{enumerate}
It is easy to see that the first two types of queries can be implemented in $\mathcal{O}(\log m)$. Consider the first type: first we use the previously computed array to compute the longest prefix of some break which is a suffix of $p[1..b]$, where $v_{i}=p[a..b]$. Then, either such prefix is
fully within $p[a..b]$, and we return it, or we need to compute its sufficiently short suffix which is a prefix of some break. In other word, we need
to locate its lowest ancestor in the prefix tree which corresponds to a prefix of length at most $|v_{i}|$. This can be done using binary search over all ancestors. Now consider
the second type: we binary search in the sorted array containing the suffixes of all $b_{j}$'s to find the one sharing the longest prefix with $v_{j}$. Finally, consider the third type of query. Because all breaks in our collection have the same length, we can find all their occurrences in $v_i$ using the array storing the leftmost occurrence in any
suffix of $p$ spending just $\mathcal{O}(1)$ time per occurrence. This implies that our running time is $\mathcal{O}(l \log m + occ)$.
\qed
\end{proof}

We now describe an algorithm for patterns with at least $2k$ disjoint $z$-breaks.

\begin{theorem}\label{theorem:algorithm_nonperiodic}
Suppose the pattern contains at least $2k$ disjoint $z$-breaks. Then pattern matching with $k$ mismatches in pc-strings can be solved in $\mathcal{O}(k\log m+ \frac{km}{z})$ time.
\end{theorem}

\begin{proof}
First we find $2k$ disjoint $z$-breaks in the pattern. We want now to detect the at most $\mathcal{O}(\frac{m}{z})$ positions in $f$ where p can potentially match. Proceeding as in the proof of Lemma~\ref{lemma:sparsify_matches}, first choose some $2k$ disjoint $z$-breaks and find all their matches in $f$ using the algorithm from Lemma~\ref{lemma:break_matching}. This costs us $\mathcal{O}(l\log m +occ)=\mathcal{O}(k\log m +\frac{km}{z})$ time. The marking phase can be done in $\mathcal{O}(\frac{km}{z})$ time. Now for each of the $\mathcal{O}(\frac{m}{z})$ positions verify whether $p$ matches there in $\mathcal{O}(k)$ time. So we can find all matches of $p$ in $f$ in $\mathcal{O}(k\log m +\frac{km}{z})$ time.
\qed
\end{proof}

\begin{remark}
The marking phase uses an array of size $\mathcal{O}(m)$. The array is reused whenever we apply the above lemma, so it adds just $\mathcal{O}(m)$ to the final space complexity.
Observe that we cannot afford to initialize the whole array in every application, as $\mathcal{O}(m)$ might actually be larger than $\mathcal{O}(k\log m+\frac{km}{z})$.
We initialize it just once in the very beginning of the whole algorithm, and during each marking phase we prepare a list of modified entries in the array. Then we clean up
just the corresponding part of the array in $\mathcal{O}(\frac{km}{z})$ time.
\end{remark}

Note that taking big $z$ makes our algorithm really fast. However, the larger is $z$, the harder is for the pattern to contain many $z$-breaks. Furthermore, we cannot expect each pattern to have many $z$-breaks, even for small $z$. Therefore, we need a different algorithm for the case when $p$ has few breaks, or is {\it highly periodic}. The algorithm has to take advantage of the regular structure of the pattern.

\section{Basic algorithm for highly periodic patterns}

\label{section:highly}
In this section we assume the pattern is highly periodic. This means we can write it in the form $p=s_1b_1s_2b_2...s_rb_rs_{r+1}$, where $r<2k$, each $b_i$ is a $z$-break and each $s_i$ is a (possibly empty) string with period at most $\frac{z}{2}$. The fragments $s_1,s_2,...,s_{r+1}$ are called {\it periodic stretches}.  As in the previous section we are interested in finding all matches (with at most $k$ mismatches) of $p$ in a pc-string $f$. 

Below we describe how to reduce the general case to the one where the number of breaks in the text is small. A very similar reasoning will be also used later in matching with errors, the only change being increasing some constants.

\begin{lemma}\label{lemma:discarding_breaks}
Suppose $f$ is a string of length at most $2m$ and $p$ is a pattern containing at most $2k$ disjoint $z$-breaks. There exists a subword $f'$ of $f$ having at most $6k+1$ disjoint $z$-breaks such that each match of $p$ in $f$ lies fully within $f'$.
\end{lemma}
\begin{proof}
Split $f=f_lf_r$ so that $|f_l|,|f_r|\leq m$. Let $f_l'$ be the shortest suffix of $f_l$ having exactly $3k$ disjoint $z$-breaks (or the whole $f_l$ in case there is no such suffix). Let $f_r'$ be the shortest prefix of $f_r$ having exactly $3k$ disjoint $z$-breaks (or the whole $f_r$ in case there is no such prefix). We define $f':=f_l'f_r'$. It is easy to verify that $f'$ has at most $6k+1$ disjoint $z$-breaks. Assume for the sake of contradiction that some match of $p$ doesn't lie within $f'$, for instance it ends beyond the right end of $f'$. In such a case $f_r'$ lies fully within this match, which means at least $2k$ out of its $3k$ breaks have to match exactly. Consequently, $p$ has at least $2k$ disjoint $z$-breaks, contradiction.
\qed
\end{proof}

\begin{proposition}\label{proposition:algorithm_text_breaks}
Suppose $f$ is a pc-string, then we can find the corresponding $f'$ from Lemma~\ref{lemma:discarding_breaks} in $\mathcal{O}(kz)$ time.
\end{proposition}
\begin{proof}
Concentrate for example on finding $f_r'$ (using notation from the proof of Lemma~\ref{lemma:discarding_breaks}). We want to simulate the algorithm \proc{Find-breaks} until the $3k$-th break is found. The method from Lemma~\ref{lemma:finding_breaks} would give us $\mathcal{O}(m)$ time, which is not good enough, but we can improve it since $f$ is a pc-string. 

We will show how to find the first break (if any) in $f_r$, the next ones are determined in the same way. We take the prefix $s'$ of $f_r$ of length $z$, and determine its period $q$ in $\mathcal{O}(z)$ time using the standard algorithm. If $q$ exceeds $\frac{z}{2}$ then $s'$ is a break and we are done. If not, the situation is as follows: we are given a string $f_r$ represented as a concatenation of $\mathcal{O}(k)$ factors of the pattern, and we want to find its unique prefix $s$, such that its length-$z$ suffix is a $z$-break and $s[1..|s|-1]$ has period $q$.  We need to find out how long $s$ is, or in other words we need to compute how far the period $q$ extends. It turns out that there is a simple formula for this, namely $|s|-1=\LCPref(f_r,f_r[q+1..|f_r|])$. We can answer such a query in time proportional to the number of blocks in the result. 

Locating a single break using the above method can take even up to $\mathcal{O}(k)$ time, but the total complexity amortizes to $\mathcal{O}(k)$, because we always cut off the processed prefix of $f_r$ and then work with the remaining part. So in total we spend $\mathcal{O}(k\cdot z + k)=\mathcal{O}(kz)$ time. 
\qed
\end{proof}

By the discussion above we can restrict ourselves to pc-strings having at most $\mathcal{O}(k)$ disjoint $z$-breaks. We will give now an algorithm achieving $\mathcal{O}(zk^4)$ running time for pattern matching with $k$ mismatches in such pc-strings. While this is not the best algorithm we have obtained, it serves well as an introduction to the more complicated $\mathcal{O}(zk^3)$ algorithm presented in the next section. 

Let us summarize the situation. We are given a pattern of the form $p=s_1b_1...s_rb_rs_{r+1}$ and a pc-string $f=s'_1b'_1...s'_{q}b'_{q}s'_{q+1}$, where $r,q=\mathcal{O}(k)$, $b$'s denote $z$-breaks and $s$'s are strings with periods not exceeding $\frac{z}{2}$. We will soon see that alignments of the pattern, for which the pattern breaks and text breaks are not too close from each other, are nice to work with. So we want to handle the remaining alignments separately.

\begin{proposition}\label{proposition:close_breaks}
There are at most $\mathcal{O}(zk^3)$ alignments of the pattern in the text such that some text break (or text endpoint) is within a distance of $z(k+1)$ from some pattern break (or pattern endpoint).
\end{proposition}

\begin{proof}
We treat pattern endpoints  and text endpoints as breaks for simplicity. There are $\mathcal{O}(k^2)$ pairs (pattern break, text break). Each such pair can violate the rule on at most $2z(k+2)$ positions. This gives us $\mathcal{O}(zk^3)$ positions in total.
\qed
\end{proof}

In this (simple) version of the algorithm we just verify all these $\mathcal{O}(zk^3)$ positions in $\mathcal{O}(k)$ time per one. This results in $\mathcal{O}(zk^4)$ complexity and leaves us with the convenient case, where all distances between pattern and text breaks (or endpoints) are at least $z(k+1)$. We call such alignments {\it fine}, and we will
soon see that if such a fine alignment results in a match, it forces all periodic stretches involved in the match to have the same period.

Starting from now we assume that the distances between consecutive breaks in the text (pattern) are at least $z(k+1)$, and otherwise group some breaks together so that the groups meet this condition. Our argument will work also for such groups but it is simpler to describe it just for breaks. 
Similarly, we want to assume that $s_1$ (and $s_{r+1}$) is either empty or has length at least $z(k+1)$, so if $0<|s_1|<z(k+1)$ we extend the first break to the left so that $s_1$ becomes empty (and do the same with $s_{r+1}$). 

One can easily see that there are at most $\mathcal{O}(k^2)$ intervals of consecutive fine alignments in the text. Within such an interval the order of appearance of the breaks does not change. Fix one interval and suppose we have at least one match there. We want to argue that in such a case all periodic stretches involved in this match are compatible, meaning
that their canonical periods are identical, and moreover start with the same offset modulo the period.

\begin{proposition}\label{proposition:periodic_mismatches}
Suppose $w_1, w_2$ are periodic strings with periods not exceeding $\frac{z}{2}$. If $w_1\neq w_2$ and $|w_1|=|w_2|\geq z(k+1)$ then there are at least $k+1$ mismatches between these two words.
\end{proposition}

\begin{proof}
We have two cases depending on whether the periods of $w_1, w_2$ are equal or not. Let $q_1$ be the period of $w_1$ and $q_2$ the period of $w_2$. If $q_1=q_2$ then in each fragment of length $q_1$ we have at least one mismatch, so at least $2(k+1)$ mismatches in total. Assume now $q_1\neq q_2$, then each fragment of length $q_1+q_2$ contains a mismatch (if not then the periodicity lemma gives a contradiction). Since $q_1+q_2\leq z$, there are at least $k+1$ mismatches.
\qed
\end{proof}

\begin{figure}[t]
\includegraphics[width=\textwidth]{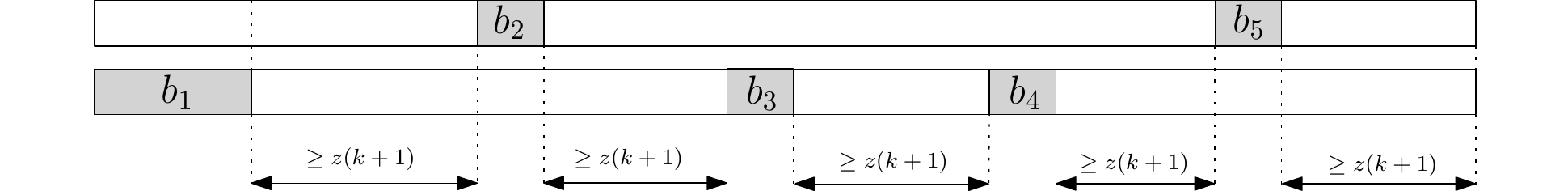}
\caption{Long overlaps between stretches imply their canonical periods are the same.}
\label{figure:overlapping stretches}
\end{figure}

Suppose there is a match at some fine alignment. Between two consecutive breaks (we consider here all pattern and text breaks) there is always a periodic portion of length at least $z(k+1)$. From Proposition~\ref{proposition:periodic_mismatches}, there must be a perfect match between the corresponding fragments. So in particular, the periods of the corresponding pattern periodic stretch and text periodic stretch agree. Considering the way how the stretches overlap each other, see Figure~\ref{figure:overlapping stretches}, one can deduce by transitivity that all periodic stretches involved in this match have the same period (they even have the same canonical period). Surprisingly, it means that if some two periodic stretches in the pattern have different canonical periods then there is no hope for matches at fine positions. 

Suppose now all the periodic stretches in the pattern have the same canonical period $u$. We consider an interval of consecutive fine alignments. Assume there is a match somewhere in this interval. One can see that each two alignments $i$ and $i+|u|$ from the interval have the same number of mismatches, because each break is aligned with a $u$-periodic stretch, so 
the fragment we compare it to is the same.
So in order to find all matches within one interval, we only need to verify at most $|u|\leq \frac{z}{2}$ alignments. Each verification takes $\mathcal{O}(k)$ time, so the time taken over all intervals is $\mathcal{O}(k^2\cdot \frac{z}{2} \cdot k)=\mathcal{O}(zk^3)$.

\begin{theorem}\label{theorem:algorithm_highlyperiodic}
For highly periodic patterns, pattern matching with $k$ mismatches in pc-strings can be solved in $\mathcal{O}(zk^4)$ time.
\end{theorem}

\section{Faster algorithm for highly periodic patterns}
\label{section:faster}

The purpose of this section is to show a faster algorithm for pattern matching with $k$ mismatches in pc-strings, assuming the pattern is highly periodic. We will improve the time complexity from $\mathcal{O}(zk^4)$ to $\mathcal{O}(zk^3)$. We will make sure that the additional space required by the improved algorithm is just $\mathcal{O}(zk^2)$, which will
be crucial in achieving linear space usage of the whole solution.

In the previous section we showed that one can assume that the text has at most $\mathcal{O}(k)$ disjoint $z$-breaks. The idea of the basic algorithm was to first work with the ``bad" alignments. An alignment was considered ``bad" if there was a text break and a pattern break close to each other (within a distance of $z(k+1)$). We took all such alignments and verified them in $\mathcal{O}(k)$ time each. The fine alignments (meaning not ``bad") were analyzed in total time $\mathcal{O}(zk^3)$. This approach, although simple, seems to be very naive. Each time there is a single pair of close breaks, we waste $\Omega(k)$ time to deal with such an alignment. It turns out that we can verify a ``bad" position in time proportional to the number of ``bad" breaks. In the following definitions and lemmas we make the idea formal.

\begin{definition}
In a fixed alignment of the pattern in the text, we call a pattern break black if there is some text break or text endpoint within distance $23zk$ from it. Similarly, we call a text break black if there is some pattern break or pattern endpoint within distance $23zk$ from it. Non-black breaks are called white.
\end{definition}

Note that one extreme case when a break is black is when it overlaps with some other break. It is convenient to deal with such situations separately. There are only $\mathcal{O}(zk^2)$ such alignments, so they can be all verified in $\mathcal{O}(zk^3)$ time, and from now on we consider only alignments where no two breaks overlap. Moreover, we can restrict our attention to alignments with at least one black break. The rest are among the fine alignments, which can be processed as shown in the previous section.

\begin{lemma}\label{lemma:fast_alignments}
After $\mathcal{O}(zk^3)$ time preprocessing, given an alignment with $B\geq 1$ black breaks we can test whether it corresponds to a match in $\mathcal{O}(B)$ time.
\end{lemma}
We will prove the above lemma in the remaining part of this section. Suppose for a moment it holds, and consider all alignments with some black breaks. Call the number of black breaks in these alignments $B_1, B_2, ..., B_g$. Then by the above lemma, each single alignment can be processed in $\mathcal{O}(B_i)$ time, so the total time is $\mathcal{O}(\sum_{i=1}^{g}B_i)$. Every specific break is black at most $\mathcal{O}(k\cdot (46z(k+1)+2z))=\mathcal{O}(zk^2)$ times, so $\mathcal{O}(\sum_{i=1}^{g}B_i)=\mathcal{O}(zk^3)$. So if we use this method to process the alignments, we will obtain an algorithm with $\mathcal{O}(zk^3)$ running time.

The main idea in the proof of the lemma is to partition the alignment into disjoint parts, such that in each of these parts we can count the number of mismatches easily. More precisely, if there are $B$ black breaks in the considered alignment, we distinguish $\mathcal{O}(B)$ intervals where the Hamming distance can be determined in $\mathcal{O}(1)$ time. For this, we need some results memorized in arrays. We will give now the details by analyzing some cases of the relative arrangement of black and white breaks. Recall we have already reduced the situation to the case where no two breaks overlap.

Consider a periodic stretch $s$ between two breaks in the pattern (text). It can be written in the form $s=u_1u^iu_2$ where $u$ is its canonical period (of length at most $\frac{z}{2}$), $i\geq 0$, $u_1$ is some suffix of $u$ and $u_2$ is some prefix of $u$. Note also that the word $u$ is primitive in such a case. It is easier to imagine the whole picture (and also to describe it) if $u_1=u_2=\varepsilon$, in other words when $s$ is a power of its canonical period. We can achieve it by merging $u_1$ ($u_2$ respectively) to the neighbouring break on the left (on the right). After this operation the breaks have lengths between $z$ and $2z$ and all periodic stretches, maybe except these at the start and at the end of the word, are powers of primitive words.

Let us fix an alignment with at least one black break, and take any black pattern break (the reasoning for text breaks is the same). We want to count the number of mismatches between it and the corresponding periodic stretch from the text. To answer such a query in constant time, we build in the preprocessing phase a table with all results. For each pattern break and periodic stretch $u^{i}$ from the text we count mismatches between the break and the stretch for every possible shift smaller than $|u|\leq\frac{z}{2}$. Each such count can be performed in $\mathcal{O}(k)$ time, which results in $\mathcal{O}(zk^3)$ time preprocessing.

\begin{figure}[t]
\includegraphics[width=\textwidth]{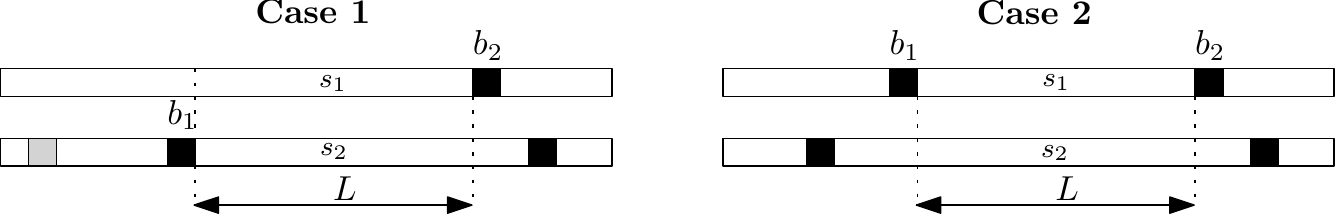}
\caption{Two consecutive black breaks.}
\label{figure:black_breaks_1}
\end{figure}

Now take two consecutive black breaks $b_1, b_2$. Consider the case, when there are no more breaks between them (of course there are no black ones, because we chose $b_1, b_2$ to be consecutive, but some white breaks might be there). Two possible situations are depicted in Figure~\ref{figure:black_breaks_1}. 
Our aim is now to count the number of mismatches between $s_1$ and $s_2$, which are length-$L$ subwords of periodic stretches from the text and pattern, respectively. If $L\geq z(k+1)$ then by Proposition~\ref{proposition:periodic_mismatches} either there are no mismatches between $s_1$ and $s_2$, or there are at least $k+1$ of them. It is easy to detect which case occurs: the strings agree if and only if their canonical periods are the same and they start with the same period offset, which can be determined in $\mathcal{O}(1)$ time after some straightforward preprocessing. So we can assume $L<z(k+1)$. We consider the cases from the Figure~\ref{figure:black_breaks_1} separately.

\begin{figure}[t]
\centering
\includegraphics[width=\textwidth]{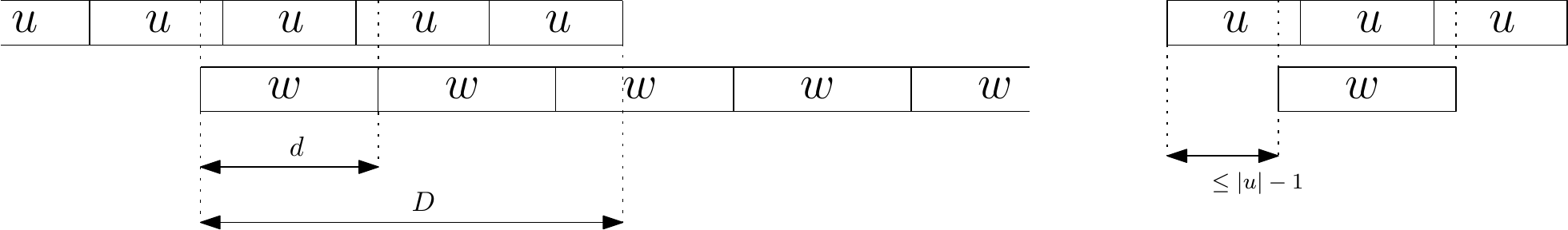}
\caption{Precomputation for a pair of periodic stretches.}
\label{figure:stretches_overlap}
\end{figure}

\begin{mycases}[listparindent=15pt]
\item In this case $s_1$ is length-$L$ suffix of some text periodic stretch, $s_2$ is length-$L$ prefix of some pattern periodic stretch. We want to precalculate all possible $\mathcal{O}(zk^3)$ results of such queries. Doing it as usually in $\mathcal{O}(k)$ time per one is unfortunately too slow. Fix one pair of periodic stretches. We will calculate all the $\mathcal{O}(zk)$ required numbers in $\mathcal{O}(zk)$ total time. Let $w$ be the canonical period of $s_1$, $d=|w|$ and let $u$ be the canonical period of $s_2$. First calculate the answer for all overlaps of length at most $d$ in $\mathcal{O}(dk)=\mathcal{O}(zk)$ time. Now to process an overlap of length $D>d$, we first use the result for $D-d$, see Figure~\ref{figure:stretches_overlap}. Then we only need to take into account the prefix of length $d$, or in other words add the number of mismatches between $w$ and some factor of an infinite word $u^{\infty}$. There are just $|u|$ essentially different factors as far as counting mismatches is concerned, see Figure~\ref{figure:stretches_overlap}, so precomputing all these $|u|$ numbers takes $\mathcal{O}(zk)$ time.

By the above discussion, we can precalculate all the required numbers in $\mathcal{O}(zk^3)$ time, which is fast enough, but unfortunately the space usage of $\mathcal{O}(zk^3)$ is 
too high for the purpose of achieving linear total space complexity. We will now explain how to decrease it to $\mathcal{O}(zk^2)$, i.e., $\mathcal{O}(z)$ per a pair of stretches,
while keeping the constant retrieval time. As previously, we compute and explicitly store the answers for all overlaps of length at most $d$, which takes $\mathcal{O}(z)$ space.
Now consider an overlap $D>d$. The overlap can be partitioned into two pieces: the part of length
$\alpha=d\lfloor\frac{D}{d}\rfloor$, and the remaining short part of length $D\bmod d$. The answer for the latter we have precomputed, so we only need how to compute the number of 
mismatches in the former, or in other words between a word of the form $w^\alpha$, and a word of the form $u^\infty$, with some shift $\beta < |u|$.
Denote by $\mism(\beta)$ the number of mismatches between $w$ and $u^\infty$ aligned with a shift $\beta <|u|$. Then the sought number
is equal to
$$\mism(\beta)+\mism((\beta + d)\bmod |u|)+... +\mism((\beta + \alpha d)\bmod |u|).$$
We need to evaluate such sum in constant time. For this we can arrange all $\mism(\beta)$ in a single array, and augment the array with all its prefix sums.
In more detail, we consider all $\beta$ with the same remainder modulo $\gcd(d,|u|)$ separately. If the remainder is $r$, we put $\mism(r)$, $\mism((r+d)\bmod |u|)$,
$\mism((r+2d)\bmod |u|)$, and so on in the array, and then compute the prefix sums. Then, knowing where $\mism(\beta)$ appears in the array, we can compute the
sought sum in constant time. For each pair of stretches we need to store a constant number of arrays of size $\mathcal{O}(z)$, so the total space complexity is as required.

\item In this case $s_1$ is a complete periodic stretch, and $s_2$ is a factor of a periodic stretch. Note that if $s_2$ has period $d$ then there are only $d$ essentially different alignments of such form. Overall there are only $\mathcal{O}(zk^2)$ possible queries, so we precalculate all of them in $\mathcal{O}(zk^3)$ time.
\end{mycases}

Then we need to consider the more general situation when there are some white breaks between two consecutive black breaks $b_1$, $b_2$. Consider two cases.

\begin{figure}[t]
\includegraphics[width=\textwidth]{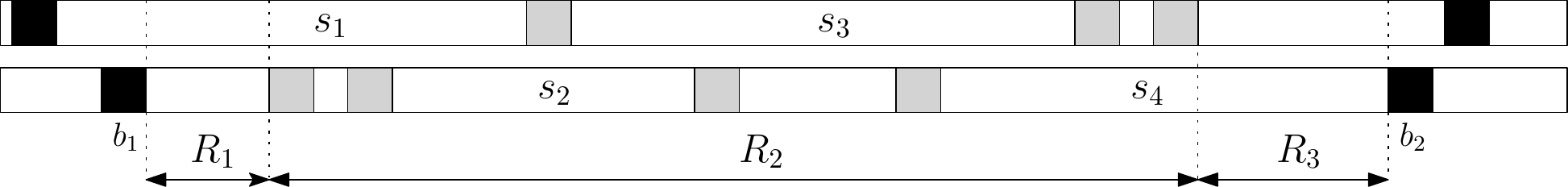}
\caption{Black breaks are represented as black boxes, white breaks are represented as grey boxes.}
\label{figure:black_breaks_2}
\vspace{0.2cm}
\includegraphics[width=\textwidth]{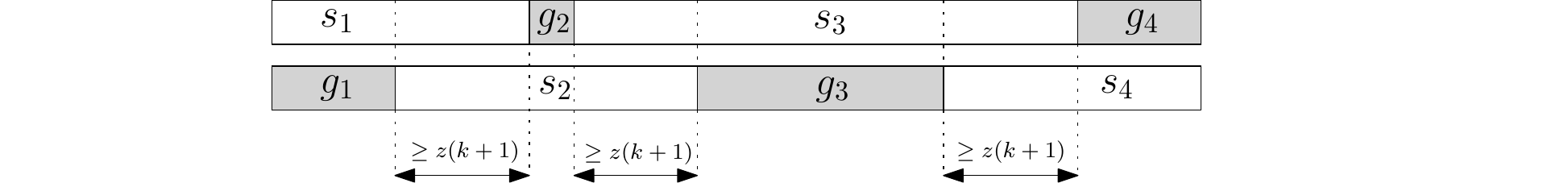}
\caption{Groups of white breaks are depicted as grey rectangles $g_{1},g_{2},g_{3},g_{4}$.}
\label{figure:white_breaks_1}
\vspace{0.2cm}
\includegraphics[width=\textwidth]{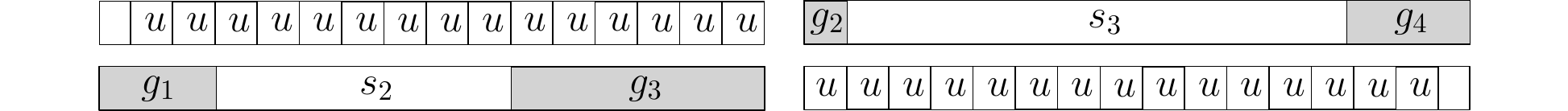}
\caption{Hamming distances that need to be precomputed. Note that $|g_2|$ and $|g_3|$ are multiples of $|u|$, but the same is not necessarily true for $|g_1|$ and $|g_4|$.}
\label{figure:white_breaks_2}
\end{figure}

\begin{mycases}[listparindent=15pt]
\item Suppose among the white breaks between $b_1$ and $b_2$ there are both pattern breaks and text breaks. An example of such a situation is depicted in Figure~\ref{figure:black_breaks_2}, where one can see three regions $R_1,R_2,R_3$.
The numbers of mismatches in regions $R_1$ and $R_3$ can be calculated in $\mathcal{O}(1)$ time by the methods explained previously. Therefore, we concentrate on $R_2$, which is the smallest region containing all white breaks between $b_1$ and $b_2$. The next step is to organize the white breaks into groups, where each group is a maximal set of consecutive white breaks of the same origin (meaning pattern break or text break),
see Figure~\ref{figure:white_breaks_1}.
Note that no two consecutive groups are of the same type, so the distances between them are at least $23zk\geq z(k+1)$. By Proposition~\ref{proposition:periodic_mismatches}, if the current alignment is a match, the overlap between $s_1$ and $s_2$ matches exactly, the same applies for $s_2, s_3$ and $s_3, s_4$. Hence, assuming there is a match at the current alignment, all $s_1, \ldots, s_4$ have the same canonical period $u$. Moreover, because all periodic stretches are power of the same primitive word, the lengths of $g_2$ and $g_3$ must be multiplies of $|u|$.
We can easily verify these conditions using some additional precomputed data. If the conditions are satisfied, then the total number of mismatches in $R_2$ is equal to the Hamming distance between $g_{2}s_{3}g_{4}$ and $uuu...$ plus the Hamming distance between $g_{1}s_{2}g_{3}$ and $...uuu$, where $uuu...$ ($...uuu$) denotes the $u$-periodic word with the appropriate length starting (ending) with $u$, see Figure~\ref{figure:white_breaks_2}. In the general case we need to know for any fragment of the pattern or the text starting and ending with a break its distance to $...uuu$ and $uuu...$. For each such fragment the only two interesting candidates for $u$ are the canonical period of the stretch preceding and the canonical period of the stretch succeeding the first group in the fragment, so we can simply store both results and always use the appropriate one. The whole preprocessing costs us $\mathcal{O}(k)$ time per fragment, so $\mathcal{O}(k^3)$ time in total.

\item All white breaks between $b_1$ and $b_2$ are of the same type, say for example they are all text breaks. Similarly as in the previous case, we will argue that it is enough to precompute the Hamming distances of text fragments with some $u$-periodic strings. We start with locating the rightmost black pattern break on the left of all our white breaks, and the leftmost black pattern break on the right. Call the region between them $R$, and the region containing all white breaks $R_2$, see Figure~\ref{figure:white_breaks_3}. As in the previous case, the nontrivial part is counting the number of mismatches within $R_{2}$ in constant time. Let $g$ be the fragment of the text corresponding to $R_2$, and $u$ the canonical period of the periodic stretch corresponding to $R$ in the pattern. As previously, we want to precompute the number of mismatches between $g$ and some $u$-periodic string. As opposed to the previous case, now it is not trivial to find just a few interesting candidates for $u$ in the preprocessing phase, and we cannot afford to perform the precomputation for all possible canonical periods. Nevertheless, it is possible to find for each fragment a unique potential candidate for $u$ using a somewhat more complex reasoning, as shown below.

\begin{figure}[t]
\includegraphics[width=\textwidth]{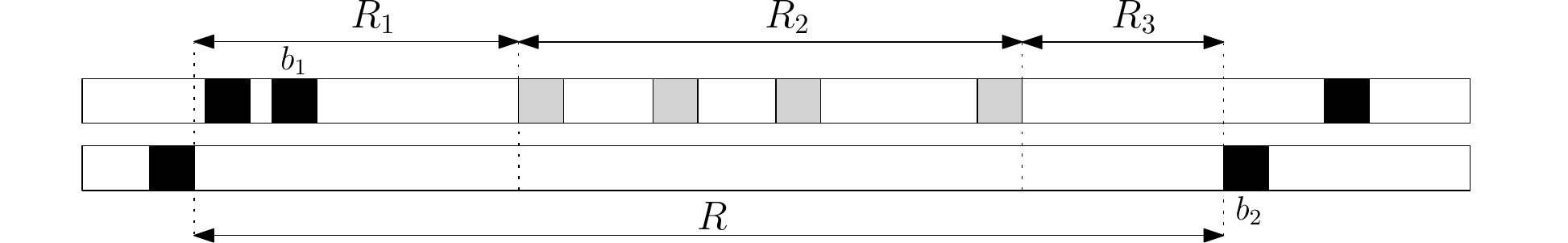}
\caption{There may be many text breaks in $R_1$ and $R_3$.}
\label{figure:white_breaks_3}
\end{figure}

Define a new region $R'$ to be $R_2$ extended by exactly $23zk$ on both sides (note it is still fully contained inside $R$), and suppose the pattern matches at the current alignment. For this to happen, necessarily most of the periodic stretches in the text within $R'$ should have the same canonical period $u$, which is formalized below.
\begin{proposition}
\label{proposition:most periods}
The total length of the periodic stretches in the text within $R'$, whose canonical period is not $u$, doesn't exceed $z(7k+2)$. 
\end{proposition}
\begin{proof}
A periodic fragment of length $\ell$ whose canonical period is other than $u$ generates at least $\lfloor \frac{\ell}{z} \rfloor$ mismatches (because in every fragment of length $z$ there is at least one mismatch, see the proof of Proposition~\ref{proposition:periodic_mismatches}).
Let the lengths of the stretches with a wrong canonical period be $\ell_{1},\ell_{2},\ldots,\ell_{s}$, and assume that $\sum_{i}\ell_{i} > z(7k+2)$. Then the number of mismatches is at least $\sum_{i}\left\lfloor\frac{\ell_{i}}{z}\right\rfloor \geq \sum_{i}\frac{\ell_{i}-z+1}{z}>7k+2-s$.  Because the whole text contains at most $6k+1$ breaks, $s\leq 6k+2$, so there are at least $k+1$ mismatches.
\qed
\end{proof}
\begin{definition}
$\TotalLen(w)$ is the total length of the periodic stretches in the text within $R'$ having canonical period $w$.
\end{definition} 
Note that the length of $R'$ exceeds $46zk$. Thus, taking into account the at most $6k+1$ occurrences of breaks in region $R'$, each of length at most $2z$, and plugging in Proposition~\ref{proposition:most periods}, we obtain that: $$\TotalLen(u)\geq |R'|-z(7k+2)-2z(6k+1)\geq |R'|-23zk>\frac{|R'|}{2}.$$

Now suppose that we choose $w$ to be the canonical period of (any) periodic stretch in the text (within $R'$), which maximizes $\TotalLen(w)$, and use it instead of the (unknown in the preprocessing stage) $u$. We take each text fragment starting and ending with a break, extend it on both sides by $23zk$, find $w$ maximizing the value of $\TotalLen(w)$ (in case there is a tie, choose any of them). Then, we calculate the number of mismatches between this fragment and all $|w|$ relevant factors of $w^{\infty}$. There are $\mathcal{O}(k^2)$ fragments of such a form, each is processed in $\mathcal{O}(zk)$ time, so the total preprocessing time is $\mathcal{O}(zk^{3})$. Then to count the mismatches between $g$ and the corresponding factor of $u^{\infty}$, we first check what $w$ we have performed the precomputation for. If $w=u$, we can use the precomputed value.
If $w\neq u$, we claim that the current alignment cannot correspond to a match. Indeed, otherwise we would have $\TotalLen(u)>\frac{|R'|}{2}$, so $u$ would be the unique canonical period maximizing $\TotalLen(w)$.
\end{mycases}

By the arguments given above, whenever we have an alignment with $B$ black breaks, we may partition it into $\mathcal{O}(B)$ regions and either count the number of mismatches in each of them, or report that it exceeds $k$, in constant time, assuming a previous precomputation requiring $\mathcal{O}(zk^{3})$ time. Thus the theorem.

\begin{theorem}\label{theorem:algorithm_highlyperiodic_fast}
For highly periodic patterns, pattern matching with $k$ mismatches in pc-strings can be solved in $\mathcal{O}(zk^3)$ time using $\mathcal{O}(zk^2)$ additional space.
\end{theorem}

It is now a good moment to specify $z$. Let $z=\frac{\sqrt{m}}{k}$. Using Theorem~\ref{theorem:algorithm_nonperiodic} and Theorem~\ref{theorem:algorithm_highlyperiodic_fast} we see that such a choice of $z$ gives us a running time $\mathcal{O}(k\log m+\sqrt{m}k^{2})=\mathcal{O}(\sqrt{m}k^{2})$ for pattern matching with $k$ mismatches in pc-strings. The additional space needed is $\mathcal{O}(zk^2)=\mathcal{O}(\sqrt{m}k)$.

\begin{remark}
Recall we assumed earlier that $z\geq 3$. Thus we should check, whether the choice $z=\sqrt{\frac{m}{k^2}}$ fulfils such requirement. Unfortunately, it doesn't whenever $m<9k^2$. But in such a case we can uncompress the pc-string and use the $\mathcal{O}(mk)$ algorithm~\cite{LandauMismatches}, also obtaining $\mathcal{O}(mk)=\mathcal{O}(\sqrt{m}\sqrt{k^2}k)=\mathcal{O}(\sqrt{m}k^{2})$ running time.
\end{remark}

By Theorem~\ref{theorem:reduce_to_PCstrings} we obtain automatically that pattern matching with $k$ mismatches in LZW-compressed text can be solved in $\mathcal{O}(nk\log^{2}m+m+n\sqrt{m}k^{2})$, which is $\mathcal{O}(n\sqrt{m}k^{2})$ because $n\geq\sqrt{m}$. The space complexity is $\mathcal{O}(n+m+\sqrt{m}k)$. This is bounded by $\mathcal{O}(n+m)$ whenever $k=\mathcal{O}(\sqrt{m})$. In the opposite case we use the $\mathcal{O}(mk)$ algorithm~\cite{LandauMismatches} to process each pc-string using
$\mathcal{O}(n+m)$ space and $\mathcal{O}(nmk)=\mathcal{O}(n\sqrt{m}k^2)$ total time.

\begin{theorem}
Pattern matching with $k$ mismatches in LZW-compressed strings can be solved in $\mathcal{O}(n\sqrt{m}k^{2})$ time and $\mathcal{O}(n+m)$ space.
\end{theorem}

\section{Algorithm for pattern matching with errors}
\label{section:errors}
In this section we discuss the algorithm for pattern matching with $k$ errors in pc-strings. It is obtained by combining our methods for compressed strings (applied for pattern matching with mismatches) with the ideas used by Cole and Hariharan~\cite{ColeHariharan}. In order to explain how to combine them, we need to restate most of their method. We skip the proofs whenever they don't require any modification to work in our setting.

In the following text, to simplify the description, we often say ``find the edit distance between $s_1$ and $s_2$". By this we actually mean ``find the exact value of the edit distance between $s_1$ and $s_2$ if it does not exceed $k$,  or report that it is greater than $k$ in the opposite case". We also often say ``matches at a given position'', which actually means ``there is a subword of the text ending at a given position such that the edit distance between it and the pattern is at most $k$''.

We start with a lemma which allows us to efficiently verify a potential match.

\begin{lemma}[see Section 5 of~\cite{ColeHariharan}]\label{lemma:verify_match}
Suppose $t$ is a text and $p$ is a pattern (given in some form), and assume that it is possible to answer $\LCSuf$ queries between their factors in constant time.
\begin{enumerate}
\item{We can verify whether $p$ matches (with at most $k$ errors) at a given position in $\mathcal{O}(k^2)$ time.}
\item{We can find all matches ending in a given window of length $l$ in $\mathcal{O}(k(k+l))$ time.}
\item{Given $t'$ (a subword of $t$) and $p'$ (a subword of $p$), we can calculate the edit distances between all $k$ longest suffixes of $t'$ and all $k$ longest suffixes of $p'$ in $\mathcal{O}(k^2)$ time.}
\end{enumerate}
\end{lemma}

In our setting, the text is given as a pc-string. In such a case it is not trivial to answer a $\LCSuf$ query in constant time, as this would require a predecessor search to locate the relevant block. Hence we need to maintain the current position in the text during the computation, or more precisely the current block and the current letter there. Inspecting the proof of the above lemma gives us that this can be easily done without sacrificing the time complexity, as long as all input positions (including the endpoints of a window or a subword) are given in such form.

Now we are ready to extend our methods for pattern matching with mismatches to the current setting. Recall that in this problem we are asked to find all ending positions of matches. Nevertheless, for convenience we will concentrate on the starting positions. This is justified, because in case of pc-strings the situation is symmetric. Let us start with the simpler case of non-periodic patterns.
\begin{theorem}\label{theorem:algorithm_errors_nonperiodic}
Suppose the pattern contains at least $2k$ disjoint $z$-breaks. Then pattern matching with $k$ errors in pc-strings can be solved in $\mathcal{O}(k\log m+\frac{k^2m}{z})$ time.
\end{theorem}
\begin{proof}
We modify the reasoning from Lemma~\ref{lemma:sparsify_matches} to show that there are at most $\mathcal{O}(\frac{m}{z})$ text windows of size $k$ containing potential matches. Partition the text into disjoint fragments of length $k$. For each $z$-break $b$ we do the marking in exactly the same way as in Lemma~\ref{lemma:sparsify_matches}. Consider now a fixed length-$k$ window in the text, it is easy to see that if there is some match ending there, then the total number of marks in the window and its two neighbouring windows must be at least $k$. Since the total number of marks is only $\mathcal{O}(\frac{mk}{z})$, we conclude that there are only $\mathcal{O}(\frac{m}{z})$ length-$k$ windows containing matches. 

Now we proceed as in the proof of Theorem~\ref{theorem:algorithm_nonperiodic}. First we use Lemma~\ref{lemma:break_matching} to simulate the marking as above. This gives us in $\mathcal{O}(k\log m+\frac{mk}{z})$ time at most $\mathcal{O}(\frac{m}{z})$ length-$k$ windows where some matches could potentially begin. We verify all these windows in $\mathcal{O}(k^2)$ time per one using Lemma~\ref{lemma:verify_match}. The total running time is $\mathcal{O}(k\log m+\frac{mk^2}{z})$. 
\qed
\end{proof}

We are left with the case when the pattern has at most $2k$ disjoint $z$-breaks. We argue (as in the case of mismatches) that the text (given as a pc-string) can be trimmed in $\mathcal{O}(zk)$ time to a string having at most $\mathcal{O}(k)$ disjoint $z$-breaks without losing any matches, see Lemma~\ref{lemma:discarding_breaks} and Proposition~\ref{proposition:algorithm_text_breaks}. Thus now the problem to solve is: given a pc-string $f$ and a pattern $p$, both containing at most $\mathcal{O}(k)$ disjoint $z$-breaks, find all starting positions of matches (with at most $k$ errors) of $p$ in $f$. We will construct an algorithm solving this problem in $\mathcal{O}(zk^4)$ time.

The first step is to deal with the inconvenient situation when we have two or more breaks too close to each other in the text or pattern. For this, as in~\cite{ColeHariharan}, we use the notion of {\it intervals}. Consider a maximal group of (say, pattern) breaks such that the periodic stretch between neighbouring breaks is shorter than $2z(k+2)$. We take the substring $s$ of the pattern containing all breaks from the group and extend it on both sides so that it starts and ends with a suitable periodic fragment. More precisely, suppose on the left of $s$ there is a periodic stretch with canonical period $w$, then we extend $s$ to the left, so that the length of the added part is between $z(k+1)$ and $z(k+1)+\frac{z}{2}$, and has prefix $w$. Then the extended string starts with at least $2(k+1)$ repetitions of $w$. Do the same on the right side and call the resulting string $I$, then $I$ is an interval. Note that there are situations when the extension to the left  or to the right is impossible, because the string $s$ lies to close to some endpoint. In such a case we extend $s$ till the left (right) endpoint and say that $I$ is incomplete on the left (on the right). We process all breaks in the text in the same manner.

If the pattern is shorter than $\mathcal{O}(zk^2)$, it can happen that there is only one interval incomplete on both sides. We don't want to deal with such pathological situations, so we decompress the pc-string and apply the Landau-Vishkin algorithm~\cite{Landau}. This results in a $\mathcal{O}(mk)=\mathcal{O}(zk^{3})$ running time for this case. From now on, we assume that there are no intervals which are incomplete on both sides.

Consider now all alignments of the pattern in the text such that some text break (or some text endpoint) is within distance $2z(k+2)$ from some pattern break (or some pattern endpoint). It is easy to see, that there are at most $\mathcal{O}(k^2)$ windows of size $\mathcal{O}(zk)$ of such positions. In our algorithm, we verify all these windows in total time $\mathcal{O}(k^2\cdot k(zk+k))=\mathcal{O}(zk^4)$ using Lemma~\ref{lemma:verify_match} to process whole windows. So let us now concentrate on the remaining {\it fine} alignments. Observe that now no two intervals overlap. Moreover, if the current position corresponds to a match, all periodic stretches between the intervals in both the pattern and the text have the same canonical period $u$, see Lemma 7.1 in~\cite{ColeHariharan}. This observation greatly simplifies the calculation of the number of errors at a fine alignment. To proceed, we need one more definition.

\begin{definition}
Let $b$ be any interval and $u$ some primitive string. We define the locked edit distance of $b$ as follows.
\begin{enumerate}
\item{If $b$ is complete on both sides, then: $$\led(b)=\min \{\ed(b,u^{\alpha}):|b|-k\leq |u^{\alpha}|\leq |b|+k\}$$}
\item{If $b$ is incomplete on the right, then: $$\led(b)=\min \{\ed(b,s):s\mbox{ is a prefix of }uuu... \mbox{ and }|s|\leq |b|+k\}$$}
\item{If $b$ is incomplete on the left, then: $$\led(b)=\min \{\ed(b,s):s\mbox{ is a suffix of }...uuu \mbox{ and }|s|\leq |b|+k\}$$}
\end{enumerate}
\end{definition}
Note that we can perform one computation of $\led(b)$ in $\mathcal{O}(k^2)$ time. Assume for example that $b$ is complete on both sides. We calculate the biggest ${\alpha}$ for which $|u^{\alpha}|\leq |b|+k$. To find the answer, it is enough to compute the edit distances between $b$ and $\mathcal{O}(k)$ longest suffixes of $u^{\alpha}$ (in fact we are interested only in suffixes being multiples of $u$, but in the worst possible case $|u|=1$, so it doesn't help much) in $\mathcal{O}(k^2)$ time using Lemma~\ref{lemma:verify_match}. We calculate all the locked edit distances of all $\mathcal{O}(k)$ intervals in the preprocessing stage in $\mathcal{O}(k^3)$ time.

Our goal is now to present simple formulas for computing the number of errors at a given fine alignment. Let $b_1, b_2, ..., b_q$ be the sequence of intervals (both pattern and text) involved in this alignment in the same order as they appear. Observe that there are only $\mathcal{O}(k^2)$ possible sequences of intervals we need to consider (each corresponding to some set of consecutive fine alignments). Let us analyze two cases separately: when the first interval $b_1$ comes from the pattern, and when it comes from the text.

\begin{figure}[t]
\includegraphics[width=\textwidth]{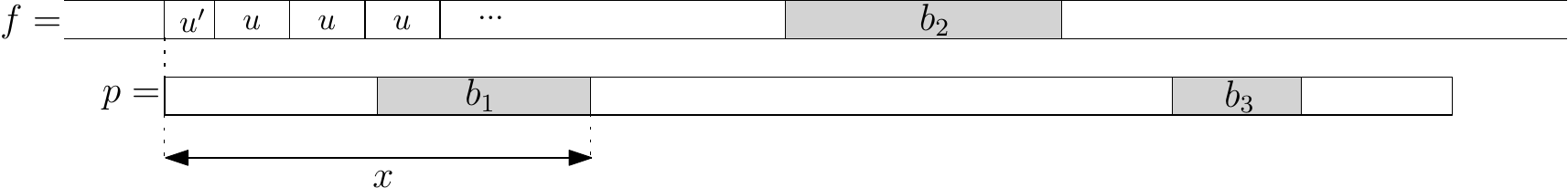}
\caption{Fine alignment where the first interval comes from the pattern.}
\label{figure:leftvalp}
\end{figure}

\begin{figure}[t]
\includegraphics[width=\textwidth]{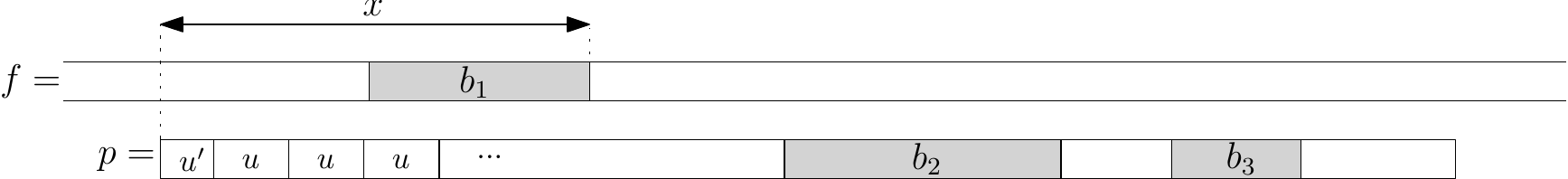}
\caption{Fine alignment where the first interval comes from the text.}
\label{figure:leftvalt}

\end{figure}

\begin{mycases}
\item{The situation is depicted in Figure~\ref{figure:leftvalp}. We denote by $u'$ the suffix of $u$ which appears at the beginning of the alignment in the text. Also let us call $x$ the prefix of the pattern ending at the last character of $b_1$. Define $\leftval_p(x)=\min\{\ed(x,u'u^{\alpha}):|x|-k\leq |u'u^{\alpha}|\leq |x|+k\} $. Cole and Hariharan showed that the number of errors at the considered alignment equals to: $$\leftval_p(x)+\led(b_2)+\led(b_3)+...+\led(b_q).$$
Thus we need to calculate $\leftval_p(x)$, which can be done in $\mathcal{O}(k^2)$ in a similar way as computing locked edit distances. Note also that the string $x$ doesn't change at different alignments, so it is enough to precompute all $\leftval_p$'s for all $|u|$ different values of $u'$. We do it in $\mathcal{O}(zk^2)$ time in the preprocessing phase, so all alignments with $b_1$ coming from the pattern can be served in $\mathcal{O}(zk^2 )$ total time. 
 
 }
\item{The situation is depicted in Figure~\ref{figure:leftvalt}. We denote by $u'$ the suffix of $u$ which appears at the beginning of the pattern. Also let us call $x$ the factor of the text starting at the alignment position and ending at the last character of $b_1$. Define $\leftval_t(x)=\min\{\ed(x,u'u^{\alpha}):|x|-k\leq |u'u^{\alpha}|\leq |x|+k\} $. Cole and Hariharan showed that the number of errors at the considered alignment equals to:  $$\leftval_t(x)+\led(b_2)+\led(b_3)+...+\led(b_q).$$ 
Note that the situation now seems to be more complicated than in Case 1, because there are many possibilities for the string $x$. However, $x$ is always of the form $u''u^{\beta}b_1$ and one can see (or find in \cite{ColeHariharan}), that $\leftval_t(u''u^{\beta+1}b_1)=\leftval_t(u''u^{\beta}b_1)$ (the main reason being that $b_1$ starts with many copies of $u$, so one of them have to match exactly). This allows us to precalculate only $\leftval_t(u''b_1)$ for all possible suffixes $u''$ of $u$. There are $\mathcal{O}(k)$ candidates for $b_1$, so the precomputation of all $\leftval_t$'s can be performed in $\mathcal{O}(zk^3)$ time. Hence all alignments with $b_{1}$ coming from the text can be served in $\mathcal{O}(zk^3)$ time.

}
\end{mycases}

The above analysis shows that all fine alignments can be served in total time $\mathcal{O}(zk^3)$. So assuming the pattern is highly periodic, we have just obtained an algorithm for pattern matching with $k$ errors in pc-strings with $\mathcal{O}(zk^4)$ running time. This allows us to prove:

\begin{theorem}\label{theorem:algorithm_errors_highlyperiodic}
For highly periodic patterns, pattern matching with $k$ errors in pc-strings can be solved in $\mathcal{O}(zk^4)$ time using $\mathcal{O}(zk^2)$ additional space.
\end{theorem}

Now we can finally specify $z$. Using Theorem~\ref{theorem:algorithm_errors_nonperiodic} and Theorem~\ref{theorem:algorithm_errors_highlyperiodic} and choosing $z=\frac{\sqrt{m}}{k}$ gives us a running time of $\mathcal{O}(k\log m+\frac{k^{2}m}{z}+zk^{4})=\mathcal{O}(\sqrt{m}k^{3})$. Then by Theorem~\ref{theorem:reduce_to_PCstrings} we obtain that pattern matching with $k$ errors in LZW-compressed text can be solved in $\mathcal{O}(nk^{2}+nk\log^{2}m+m+n\sqrt{m}k^{3})=\mathcal{O}(n\sqrt{m}k^{3})$ time. The space complexity is $\mathcal{O}(n+m+\sqrt{m}{k})$, so if $k=\mathcal{O}(\sqrt{m})$, it can be bounded by $\mathcal{O}(n+m)$. If $k$ is larger, we use the $\mathcal{O}(mk)$ algorithm~\cite{Landau} to process each pc-string using $\mathcal{O}(n+m)$ space and $\mathcal{O}(nmk)=\mathcal{O}(n\sqrt{m}k^2)$ total time.

\begin{theorem}
Pattern matching with $k$ errors in LZW-compressed strings can be solved in $\mathcal{O}(n\sqrt{m}k^{3})$ time and $\mathcal{O}(n+m)$ space.
\end{theorem}

\section{Conclusions}

We constructed efficient algorithms for pattern matching with $k$ mismatches or errors in LZW-compressed strings. The main difference with the previously known solutions is that we used both a periodicity-based argument, and the repetitive structure of a compressed string, which allowed us to achieve a better running time for small values of $k$, which seems to be the most natural setting. Our solutions achieve $\mathcal{O}(n\sqrt{m}f(k))$ running time, where $f$ is some function depending only on the bound on the number of errors $k$, while the complexity of the best previously known algorithms was of the form $\mathcal{O}(nmf(k))$.

A natural question is whether our complexities can be improved. More concretely, in pattern matching with $k$ mismatches instead of considering each non-fine alignment separately, we were able to analyze them in a more global manner in Section~\ref{section:faster}. Can a similar reasoning be applied to accelerate pattern matching with $k$ errors?

Furthermore, is it possible to obtain an algorithm with running time $\mathcal{O}(nm^{\gamma}f(k))$, where $\gamma<\frac{1}{2}$, or maybe even $\gamma=0$?

\bibliographystyle{splncs03}
\bibliography{biblio}

\end{document}